\documentclass[DIV=15,11pt,abstract]{scrartcl}
\usepackage[T1]{fontenc}
\usepackage[utf8]{inputenc}
\usepackage[english]{babel}

\usepackage{lmodern}
\usepackage{microtype}

\setkomafont{title}{\large\rmfamily\bfseries}
\setkomafont{section}{\normalfont\rmfamily\bfseries\large}
\setkomafont{subsection}{\normalfont\rmfamily\bfseries}
\setkomafont{subsubsection}{\normalfont\rmfamily\bfseries\normalsize}
\setkomafont{paragraph}{\normalfont\rmfamily\bfseries}
\setkomafont{subparagraph}{\normalfont\itshape}

\RedeclareSectionCommands[    
tocentryformat=\normalfont\bfseries,
tocpagenumberformat=\normalfont\bfseries
]{section}
\RedeclareSectionCommands[    
tocentryformat=\normalfont,
tocpagenumberformat=\normalfont
]{subsection,subsubsection,paragraph,subparagraph}

\usepackage[bottom]{footmisc}   

\usepackage[intlimits,sumlimits]{amsmath}
\usepackage{amssymb, commath, mathtools, bm, relsize}
\usepackage{amsthm, thmtools}
\usepackage{dsfont}
\usepackage{tensor}

\numberwithin{equation}{section}
\usepackage{tikz}
\usetikzlibrary{positioning, calc}
\tikzset{brauer/.style={
node distance=2em,
font=\footnotesize,
every node/.style = {circle, fill=black, minimum size=.2em, inner sep=0pt},
baseline={([yshift=-.5ex]current bounding box.center)}
}}
\usepackage{graphicx} 
\usepackage{blindtext} 
\usepackage[format=plain, labelfont=bf, labelsep=period]{caption} 
\usepackage{ytableau}
\usepackage{soul}    
\usepackage[toc,page]{appendix}
\usepackage[noadjust]{cite} 
\DeclareOldFontCommand{\bf}{\normalfont\bfseries}{\mathbf}
\usepackage{xcolor}
\definecolor{myBlue}{RGB}{1,1,141}
\usepackage[colorlinks=True]{hyperref}
\hypersetup{allcolors=myBlue}
\newtheoremstyle{plain}
{}{}{\itshape}{}{\bfseries}{.}{.5em}{}
\theoremstyle{plain}    
\newtheorem{theorem}{Theorem}
\newtheorem*{theorem*}{Theorem}
\newtheorem{lem}[theorem]{Lemma}
\newtheorem{cor}[theorem]{Corollary}

\newtheorem{definition}[theorem]{Definition}

\newcommand{\cB}{\mathcal{B}}

\newcommand{\cD}{\mathcal{D}}

\newcommand{\cG}{\mathcal{G}}

\newcommand{\cS}{\mathcal{S}}

\newcommand{\fS}{\mathfrak{S}}
\newcommand{\fb}{\mathfrak{b}}

\newcommand{\vM}{\vec{M}}

\newcommand{\vcG}{\vec{\mathcal{G}}}

\DeclareMathOperator{\sgn}{sgn}

\usepackage{authblk}
\usepackage{orcidlink}
\author[1]{H. Keppler\,\protect\orcidlink{0000-0001-5877-6686}}
\affil[1]{\normalsize\itshape 
	Heidelberg University, Institut f{\"u}r Theoretische Physik, Philosophenweg 19, 69120 Heidelberg, Germany, EU
	\authorcr \hfill}
\author[2]{T. Krajewski\,\protect\orcidlink{0000-0001-8499-3000}}
\affil[2]{\normalsize\itshape Aix Marseille Univ, Université de Toulon, CNRS, CPT, Marseille, France, EU \authorcr \hfill}
\author[3]{T. Muller\,\protect\orcidlink{0000-0003-2024-8110}}
\affil[3]{\normalsize\itshape 
	Université de Bordeaux, LaBRI CNRS UMR 5800, Talence, France, EU
	\authorcr \hfill}
\author[3,4,5]{A. Tanasa\,\protect\orcidlink{0000-0003-1671-5652}}
\affil[4]{\normalsize\itshape DFT, H. Hulubei Nat. Inst. Phys. Nucl. Engineering, Magurele,  Magurele, Romania, EU \authorcr \hfill}
\affil[5]{\normalsize\itshape Université Sorbonne Paris Nord, LIPN, CNRS UMR 7030, Villetaneuse, France, EU}

\begin{document}
    \title{Duality of \boldmath$O(N)$ and $Sp(N)$ random tensor models: tensors with symmetries}
    \date{}

{\hypersetup{allcolors=black}
	\maketitle
	\begin{abstract}
	    In a recent series of papers, a duality between 
	    orthogonal and symplectic random tensor models has been proven, first for quartic models and then for models with interactions of arbitrary order. However, the tensor
	    models considered so far in the literature had  
	    no symmetry under permutation of the indices.
	    In this paper, we  generalize these results for tensors models with interactions of arbitrary order
	    which further have non-trivial symmetry under the permutation of the indices. 
	    Totally symmetric and anti-symmetric tensors are thus  treated as a particular case of our result. 
	\end{abstract}
\microtypesetup{protrusion=false}
	\setcounter{tocdepth}{2}
	\tableofcontents
\microtypesetup{protrusion=true}
}

\section{Introduction}

Random tensor models (see the recent books \cite{gurau, tanasabook} or the reviews \cite{Gurau-invitation, Guraureview, TanasaSIGMA, Tanasa2012, GurauAIHPD}) are $0-$dimensional quantum field theoretical generalisations of the celebrated matrix models \cite{DiFrancesco}. Within this framework, they can be seen as
probability measures on tensor spaces; this is the point of view we take in this paper.

Tensor models have thus been used as tools to generate discrete random geometries in more than two dimensions.
Moreover, they have been further used to construct models similar to the holographic Sachdev-Ye-Kitaev model but without quenched disorder \cite{Witten:2016iux,Klebanov}, and new (\textit{melonic}) Conformal Field Theories \cite{Giombi:2017dtl,Bulycheva:2017ilt,Giombi:2018qgp,Klebanov:2018fzb,Gurau-TFT, Harribey:2022}.

Many of the original rigorous results on tensor models relied on the presence of a very large symmetry group (usual several distinct copies of $U(N)$ or $O(N)$) that forbids the tensor to have any symmetry under permutation of their indices \cite{Gurau-N,
TanasaMO, Bonzom:2012hw,Carrozza:2015adg,sabine,sylvan,Carrozza:2021qos,Dartois_2013,krajewski2023double}.
Later on, tensor models with tensors living on some non-trivial (mostly $O(N)$) representation where studied systematically \cite{sabine, sylvan, Carrozza:2021qos, Klebanov:2017, Gurau:2017,Carrozza}).

In \cite{Duality_Hannes, kepplermuller23} the authors studied tensor models with symplectic symmetry $Sp(N)$ in which case the tensor components sometimes are anticommuting (fermionic/odd graßmann) variables. 

Relations between the representations of $O(N)$ and $Sp(N)$ have a long history. thus, King \cite{King0} showed that the dimensions of irreducible representations of both groups agree, when exchanging symmetrization and antisymmetrization (transposed Young tableau) and replacing $N$ by $-N$. So called negative dimension theorems, or $N$ to $-N$ dualities, relating the orthogonal and symplectic group vial the formal relation $SO(-N)\simeq Sp(N)$ are well known \cite{King,Cvitanovic,parisi-sourlas,ramgoolam1994,Cvitanovicbook,Mkrtchyan-Veselov}
for matrix and vector models. 
Several incarnations of this relation can be found in the literature: for even $N$, $SO(N)$ and $Sp(N)$ gauge theories are known to be related by changing $N$ to $-N$ \cite{Mkrtchian}; a vector model with symplectic fermions in three space-time dimensions has been studied in \cite{LeClair} and an example of $SO(N)$ and $Sp(N)$ gauge theories with matter fields and Yukawa interactions can be found in \cite{Litim}; a duality between orthogonal and symplectic matrix ensembles (the $\beta=1,4$ ensembles) has been shown in \cite{goe-gse}.  
 
From a supergeometric or supersymmetric point of view such relations can be seen to arise naturally \cite{Dunne:1989}. 

As a natural followup of these matrix model results, we show in this paper how the $N$ to $-N$ symmetry arises in the tensor model case for tensors with interactions of arbitrary order, which further have non-trivial symmetry under the permutation of their indices. 
This result is a  
generalization of similar results obtained in simpler settings: for quartic interactions this was proven in \cite{Duality_Hannes} and for tensor models with interactions of arbitrary order this was done in \cite{kepplermuller23}. However, let us emphasize that, unlike the results of this paper, both the results of \cite{Duality_Hannes} and \cite{kepplermuller23} were obtained for tensor models that had no symmetry under the permutation of indices.

More precisely, the main result of this paper is the following.
We consider tensors of order $D$ that transform in some tensor representation $R$ of $O(N)$ or $Sp(N)$. This  
implies that the tensors may obey some non-trivial symmetry under permutation of their indices. In order to treat models with orthogonal and symplectic symmetry simultaneously, we introduce a grading parameter $\fb\in\{0,1\}$, such that $\fb=0$ corresponds to the $O(N)$ symmetric model and $\fb=1$ to the $Sp(N)$ symmetric one. The tensor components are real fermionic (anticommuting, odd) if $\fb=1$ and $D$ is odd, and real bosonic (commuting, even) otherwise.

\begin{definition}\label{def: intro}
The real graded tensor model with symmetry $R$ is defined by the measure
\begin{equation}
\begin{gathered}
d\mu[T] \simeq e^{-S[T]}\ \prod_{a_1,\dots,a_D} dT^{a_1\dots a_D} \; ,
 \\
S[T] = T^{a_\mathcal{D}} C^{-1}_{a_\mathcal{D} b_\mathcal{D}} T^{b_\mathcal{D}} + \sum_{\substack{\mathcal{S}\ \text{connected,} \\ |V(\mathcal{S})|>2}} \frac{\lambda_{\mathcal{S}}}{|V(\mathcal{S})|/D}\, I_{\mathcal{S}}(T) \; ,
\end{gathered}
\end{equation}
where $g^{\fb}_{a_cb_c}$ is the Kronecker $\delta_{a_cb_c} $ for $\fb=0$ or the canonical symplectic form $\omega_{a_cb_c}$ for $\fb=1$ and the sum runs over independent connected invariants $I_{\cS}(T)$ of order higher than two, indexed by undirected standed graphs $\cS$ (see Section~\ref{sec: model} for more details).
\end{definition}

The partition function $Z$ and the expectation value of an invariant $\langle{I_{\cS}(T)}\rangle$ are defined by:
\begin{equation}\label{eq: ZandI}
Z(\{\lambda\})=\int d\mu[T], \quad\text{and}\quad \langle{I_{\cS}}(T)\rangle(\{\lambda\} )=\frac{1}{Z} \int d\mu[T]\ I_{\cS}(T) \; ,
\end{equation}
and can be evaluated in perturbation theory.
The main theorem of this paper is:
\begin{theorem}\label{thm: main} 
The perturbative series of the partition function $Z$ and expectation values of invariants $\langle{I_{\cB}(T)}\rangle$ can be expressed as a formal sum over 2-colored stranded graphs $\cG$.
Each summand, corresponding to a specific graph $\cG$ (called the amplitude of that graph), writes as a product:
\begin{equation}\label{eq: amplitude}
K(\{\lambda\},\cG)\cdot \big((-1)^{\fb}N\big)^{F(\cG)} \;,
\end{equation}
of a term depending on $N$ and a term $K$, encoding both the dependence on the coupling constants $\lambda_\cS$  and some combinatorial factors associated to $\cG$  (see Section~\ref{sec: model} for the relevant definitions).
\end{theorem}

The main result of this paper follows as a direct consequence of the theorem above:
\begin{cor} \label{corrollary}
Tensor models of the form in Def.~\ref{def: intro} with symmetry given by the $O(N)$ tensor representation $R$ are dual to corresponding tensor models with $Sp(N)$ symmetry given by the representation with transposed Young diagrams $R^\prime$ (exchanging symmetrization and antisymmetrization) in the sense that the amplitudes of graphs in their perturbative expansions are mapped into each other after a change of $N$ to $-N$.
\end{cor}
\begin{proof}
This follows from Theorem~\ref{thm: main}. The replacement $\fb\to \fb+1\mod 2$ and $N\to -N$ leaves the amplitude \eqref{eq: amplitude} unchanged, and, as will be noted in Section~\ref{sec: irreps}, the shift $\fb\to \fb+1\mod 2$ exchanges symmetrization and antisymmetrization in the tensor representation $R$.
This has the effect of transposing all Young diagrams $\lambda\to \lambda^\prime$, and leads to the tensor representation $R^\prime$.
\end{proof}

The paper is organized as follows. In Section~\ref{sec: Prereq} we recall several 
results on representation theory of the orthogonal and symplectic group focusing 
on the Brauer algebra, that plays a similar role as the algebra of the symmetric group for representations of $GL(N)$. 
At the end of this section we give a dictionary between notions used in the physics/tensor model and representation theory literature. In Section~\ref{sec: model} we define the tensor models of interest for this paper
and give their diagrammatic representation in terms of stranded graph. In Section~\ref{sec: proof} we give the proof of our main result, and in Section~\ref{sec: example} we use, as an explicit example, the totally symmetric and antisymmetric tensor representations to illustrate the duality between $O(N)$ and $Sp(N)$ tensor models proved in the previous section.

\section{Prerequisite}\label{sec: Prereq}
\subsection{Irreducible representations of the orthogonal and symplectic group}\label{sec: irreps}

In this section we review some definitions and results of the 
theory of irreducible representations of the general linear group $GL(N)$ and its connection to representations of the symmetric group $\fS_D$ and Young diagrams.
We further review irreducible representations of the groups $O(N)$ and $Sp(N)$, preserving some non-degenerate bilinear form and their connections to the Brauer algebra.

Let $V=\mathbb{R}^N$. Both $GL(N)$ and $\fS_D$ act on the tensor product space $V^{\otimes F}$. Irreducible representations of $GL(N)$ can be obtained as the image of certain elements of the group algebra $\mathbb{C}\fS_D$ (Young symmetrizers). Analogously, irreducible representations of $O(N)$ or $Sp(N)$ can be obtained using projectors, defined by elements of the Brauer algebra $B_D$. 

Our exposition is based on \cite{FultonHarris} and, when concerning the Brauer algebra, on \cite{braueralg}. 
We further refer the interested reader to \cite{brauer,wenzl} or to the books \cite{weyl} or  \cite{Cvitanovicbook}.

\paragraph{Young tableaux.}
For general combinatorial references on Young tableaux, we refer to Chapter $XIV$ of the handbook \cite{handbook}.
To a partition $\lambda=(\lambda_1,\lambda_2,\dots,\lambda_k)$ of $D\in\mathbb{N}$, denoted as $\lambda\vdash D$, i.e.~a sequence of non increasing integers with $|\lambda|=\sum_{i=1}^k\lambda_i=D$, we associate a \emph{Young diagram}
\begin{equation}
\lambda=\quad
\ytableausetup{centertableaux}
\begin{ytableau}
\none[\lambda_1\hspace{1em}] & & & \\
\none[\lambda_2\hspace{1em}] & & & \\
\none[\lambda_3\hspace{1em}] & & \\
\none[\lambda_4\hspace{1em}] & \\
\none[\lambda_5\hspace{1em}] & \\
\end{ytableau}
\end{equation}
with $\lambda_i$ boxes in the $i$th row. Note that we are using here the English notation for Young diagrams and tableaux.
The dual diagram $\lambda^\prime$ is obtained by interchanging rows and columns in the Young diagram. Let us {recall} that
Young diagrams can be used to define projectors onto irreducible representations of the symmetric group $\fS_D$.

Given a Young diagram, a \emph{Young tableau} is a numbering of the boxes by the integers $1,2,\dots, D$. The canonical Young tableau is obtained by numbering the boxes consecutively:
\begin{equation}
\begin{ytableau}
1 & 2 & 3 \\
4 & 5 & 6 \\
7 & 8 \\
9 \\
10 \\
\end{ytableau} \;.
\end{equation}
Define the sets of row and column permutations:
\begin{equation}
\begin{aligned}
P_\lambda &= \{ g\in\fS_D\ |\ g\ \text{preserves each row} \} \;, \\
Q_\lambda &= \{ g\in\fS_D\ |\ g\ \text{preserves each column} \} \;.
\end{aligned}
\end{equation}
Next, one introduces two elements of the group algebra $\mathbb{C}\fS_D$:
\begin{equation}
    a_\lambda = \sum_{g\in P_\lambda} g \;,\quad b_\lambda=\sum_{g\in Q_\lambda} \sgn(g) g \;.
    \label{eq:ab}
\end{equation}
Noting that $\mathbb{C}\fS_D$ acts on $V^{\otimes D}$ by permuting factors, $a_\lambda$ acts as a symmetrizer and $b_\lambda$ as an antisymmetrizer on the tensors. Finally, the \emph{Young symmetrizer} is defined as:
\begin{equation}
    c_\lambda = a_\lambda \cdot b_\lambda \;.
\end{equation}
Consider as an example $\lambda=\ytableausetup{boxsize=.5em}\ydiagram{3}$ or $\ydiagram{1,1,1}$. The image of the action of $c_\lambda$ on $V^{\otimes 3}$ is $\mathrm{Sym}^3 V$ or $\bigwedge^3 V$, the spaces of totally symmetric or antisymmetric tensors, respectively.
\ytableausetup{boxsize=normal}

The permutation group $\fS_D$ acts on tensors in $V^{\otimes D}$ by permutation of the indices, with $V$ a vector space of dimension $N$. Then all the previous projectors give rise to representations, which are in general reducible. The dimension of a representation indexed by the Young diagram $\lambda$ reads
\begin{align}
\text{dim}(\pi_{\lambda,N})=\prod_{(i,j)\in\lambda}\frac{N-i+j}{h_{ij}}
\end{align}
with $i$ (resp. $j$) the row  (resp. column) label of the box and $h_{ij}$ the hook length of the box $(i,j)$, i.e. 
\begin{align}
    h_{i,j}=\#\big\{ (k,l)\text{ with $k=i$, $l\geq j$ or $l=j$, $k\geq i$}\big\}.
\end{align}
Then, it is worthwhile to notice that it is a polynomial in $N$ that obey the relation
\begin{align}
\text{dim}(\pi_{\lambda,-N})=
(-1)^{|\lambda|}\text{dim}(\pi_{\lambda',N})
\end{align}
with $|\lambda|$ the number of boxes in $\lambda$ and $\lambda'$ the dual diagram.  Therefore, trading $N$ for $-N$ involves exchanging rows and columns, or equivalently, symmetrization and antisymmetrization. For example
\begin{align}
\text{dim}\Big(\ytableausetup{boxsize=0.5em}\ydiagram{2,1,1},N\Big)
=\frac{N(N-1)(N+1)(N+2)}{4\cdot 2\cdot 1\cdot 1}
\mathop{\leftrightarrow}\limits_{\text{}duality}
\text{dim}\Big(\ytableausetup{boxsize=0.5em}\ydiagram{3,1},-N\Big)
=\frac{N(N-1)(N-2)(N+1)}{4\cdot 2\cdot 1\cdot 1}
\end{align}
These are representations of the symmetric group but for our purposes it turns out to be helpful to identify irreducible representations of the groups $O(N)$ and $Sp(N)$ inside the previous ones, as we shall do in the following.

\paragraph{Representations.}
Let us recall that a representation of the group $GL(N)$ on $V^{\otimes D}$ is semisimple and decomposes into a direct sum of irreducible representations that are determined by irreducible representations of $\fS_D$, and thus indexed by Young diagrams. For simplicity, we focus on $N$ much larger  than $D$ ($N\geq 2D$). 
Note that for small $N$ not all Young diagrams give irreducible representations.

An analogous construction holds for the groups $O(N)$ and $Sp(N)$ that preserve a non-degenerate (skew-)symmetric bilinear form. The main difference lies in the ability to form traces by contacting two factors of $V^{\otimes D}$ with the bilinear form. To allow for these contractions, the group algebra $\mathbb{C}\fS_D$ is replaced by the Brauer algebra $B_D$ \cite{brauer}. As subgroups $O(N), Sp(N)\subset GL(N)$, irreducible representations of $GL(N)$ are still representations of $O(N)$ and $Sp(N)$, but not necessarily irreducible. However, irreducible $O(N)$ or $Sp(N)$ representations can be obtained by traceless projections of irreducible representations of $GL(N)$. In \cite{braueralg}, a universal traceless projector $\mathfrak{P}_D\in B_D$ was constructed, such that irreducible $O(N)$ or $Sp(N)$ representations can be obtained by first subtracting traces by applying $\mathfrak{P}_D$, and second applying a projector (e.g.~Young symmetrizer) to an irreducible $GL(N)$ representation. 
Note that, in particular, both operations commute.

\paragraph{Brauer algebra.}
Let us now exhibit the Brauer algebra $B_D(z)$, for $D\in \mathbb{N}$, $z\in\mathbb{C}$.

For $D\in\mathbb{N}$, draw two horizontal rows of vertices labelled $1, 2, \dots, D$. \emph{Brauer diagrams} are represented by pairings of these $2D$ vertices. If every vertex in the top row is connected to a vertex in the bottom row, these elements represent permutation diagrams. Thus, 
$\fS_D$ is a subset of the diagrams and  $\mathbb{C} \fS_D$ a subset of the algebra. 
For example:
\begin{equation}
\sigma = \begin{tikzpicture}[brauer]
	\foreach \i in {1,2,3,4}{
		\node (\i a) at (0.7*\i-.7,0) {};  \node (\i b) at (0.7*\i-0.7,-1) {};  
	    \node[above=.5ex,fill=none] at (\i a.north) {\i}; \node[below=.5ex,fill=none] at (\i b.south) {\i}; }
    \draw (1a)--(2b) (2a)--(3b) (3a)--(1b) (4a)--(4b);
\end{tikzpicture}
\;,\quad 
\tau = \begin{tikzpicture}[brauer]
	\foreach \i in {1,2,3,4}{
		\node (\i a) at (0.7*\i-.7,0) {};  \node (\i b) at (0.7*\i-0.7,-1) {};    
		\node[above=.5ex,fill=none] at (\i a.north) {\i}; \node[below=.5ex,fill=none] at (\i b.south) {\i}; }
	\draw (1a)--(2b) (2a)--(1b) (3a)--(4b) (4a)--(3b);
\end{tikzpicture} \;.
\end{equation}
For simplicity, from now one we omit the labels on our diagrams.
Since Brauer diagrams are more general than permutation diagrams, the set of Brauer diagrams includes elements such as
\begin{equation}
\beta =\ \begin{tikzpicture}[brauer]
	\foreach \i in {1,2,3,4}{
		\node (\i a) at (0.7*\i-.7,0) {};  \node (\i b) at (0.7*\i-0.7,-1) {}; }
	\draw (1a)to[bend right](3a) (2a)to[bend right](4a) (1b)to[bend left](2b) (3b)to[bend left](4b);
\end{tikzpicture} \;,\quad
\upsilon =\ \begin{tikzpicture}[brauer]
	\foreach \i in {1,2,3,4}{
		\node (\i a) at (0.7*\i-.7,0) {};  \node (\i b) at (0.7*\i-0.7,-1) {}; }
	\draw (1a)to[bend right](2a) (3a)--(4b) (4a)--(2b) (1b)to[bend left](3b);
\end{tikzpicture} \;,
\end{equation}
having arcs connecting vertices of the same row. 
The product of two Brauer diagrams $\sigma\tau$ is defined by placing $\sigma$ below $\tau$ and ``straightening'' the lines:
\begin{equation}
\sigma\tau =\ \begin{tikzpicture}[brauer]
	\foreach \i in {1,2,3,4}{
		\node (\i a) at (0.7*\i-.7,0) {};  \node (\i b) at (0.7*\i-0.7,-1) {};  \node (\i c) at (0.7*\i-0.7,-2) {};  }
	\draw (1b)--(2c) (2b)--(3c) (3b)--(1c) (4b)--(4c);
	\draw (1a)--(2b) (2a)--(1b) (3a)--(4b) (4a)--(3b);
\end{tikzpicture}
\ =\
\begin{tikzpicture}[brauer]
	\foreach \i in {1,2,3,4}{
		\node (\i a) at (0.7*\i-.7,0) {};  \node (\i b) at (0.7*\i-0.7,-1) {};  }
	\draw (1a)--(3b) (2a)--(2b) (3a)--(4b) (4a)--(1b);
\end{tikzpicture}
\;.
\end{equation}
For permutation diagrams, this is equivalent to the product of the permutations. Whenever loops appear, they get deleted to obtain again a Brauer diagram.

The Brauer algebra $B_D(z)$ is the free $\mathbb{C}$-algebra on the set of Brauer diagrams together with the above product and the additional rule stating that when $l\geq0$ loops appear in the product of two Brauer diagrams, the resulting diagram gets multiplied by a factor $z^l$.
\begin{equation}
\beta\upsilon =\ \begin{tikzpicture}[brauer]
	\foreach \i in {1,2,3,4}{
		\node (\i a) at (0.7*\i-.7,0) {};  \node (\i b) at (0.7*\i-0.7,-1) {};  \node (\i c) at (0.7*\i-0.7,-2) {};  }
	\draw (1a)to[bend right](2a) (3a)--(4b) (4a)--(2b) (1b)to[bend left](3b);
	\draw (1b)to[bend right](3b) (2b)to[bend right](4b) (1c)to[bend left](2c) (3c)to[bend left](4c);
\end{tikzpicture}
\ = z\
\begin{tikzpicture}[brauer]
	\foreach \i in {1,2,3,4}{
		\node (\i a) at (0.7*\i-.7,0) {};  \node (\i b) at (0.7*\i-0.7,-1) {}; }
	\draw (1a)to[bend right](2a) (3a)to[bend right](4a) (1b)to[bend left](2b) (3b)to[bend left](4b) ;
\end{tikzpicture}
\;.
\end{equation}
Note that one has: $\mathbb{C}\fS_D\subset B_D(z)$ (diagrams with zero arcs).

A set of generators of $B_D(z)$ is given by $\sigma_i$ and {$\beta_i$} ($i=1,2,\dots,D-1$):
\begin{equation}\label{eq: generators}
\sigma_i=\ \begin{tikzpicture}[brauer]
	\foreach \i in {1,3,4,5,6,8}{
		\node (\i a) at (0.7*\i-0.7,0) {};  \node (\i b) at (0.7*\i-0.7,-1) {};  }
	\foreach \i in {2,7}{
		\node[fill=none] at (0.7*\i-0.7,0) {$\dots$};  \node[fill=none] at (0.7*\i-0.7,-1) {$\dots$};  }
	\draw (1a)--(1b) (3a)--(3b) (4a)--(5b) (5a)--(4b) (6a)--(6b) (8a)--(8b) ;
	\node[above=.5ex,fill=none] at (4a.north) {$i$}; \node[above=-1ex,fill=none] at (5a.north) {$i+1$}; 
	\node[above=.5ex,fill=none] at (1a.north) {1}; \node[above=.5ex,fill=none] at (8a.north) {$D$};
\end{tikzpicture}
\;,\quad
\beta_i=\ \begin{tikzpicture}[brauer]
	\foreach \i in {1,3,4,5,6,8}{
		\node (\i a) at (0.6*\i-0.6,0) {};  \node (\i b) at (0.6*\i-0.6,-1) {};  }
	\foreach \i in {2,7}{
		\node[fill=none] at (0.6*\i-0.6,0) {$\dots$};  \node[fill=none] at (0.6*\i-0.6,-1) {$\dots$};  }
	\draw (1a)--(1b) (3a)--(3b) (4a)to[bend right](5a) (4b)to[bend left](5b) (6a)--(6b) (8a)--(8b) ;
	\node[above=.5ex,fill=none] at (4a.north) {$i$}; \node[above=-1ex,fill=none] at (5a.north) {$i+1$}; 
	\node[above=.5ex,fill=none] at (1a.north) {1}; \node[above=.5ex,fill=none] at (8a.north) {$D$};
\end{tikzpicture} \;.
\end{equation}
Furthermore, we introduce the following elements for $i<j$:
\begin{equation}\label{eq: betaij}
\sigma_{ij} = \ \begin{tikzpicture}[brauer]
	\foreach \i in {1,3,4,5,7,8,9,11}{
		\node (\i a) at (0.6*\i-0.6,0) {};  \node (\i b) at (0.6*\i-0.6,-1) {};  }
	\foreach \i in {2,6,10}{
		\node[fill=none] at (0.6*\i-0.6,0) {$\dots$};  \node[fill=none] at (0.6*\i-0.6,-1) {$\dots$};  }
	\draw (1a)--(1b) (3a)--(3b) (4a)--(8b) (5a)--(5b) (7a)--(7b) (8a)--(4b) (9a)--(9b) (11a)--(11b) ;
	\node[above=.5ex,fill=none] at (4a.north) {$i$}; \node[above=.5ex,fill=none] at (8a.north) {$j$}; 
	\node[above=.5ex,fill=none] at (1a.north) {1}; \node[above=.5ex,fill=none] at (11a.north) {$D$};
\end{tikzpicture}
\;,\quad
\beta_{ij} = \ \begin{tikzpicture}[brauer]
	\foreach \i in {1,3,4,5,7,8,9,11}{
		\node (\i a) at (0.6*\i-0.6,0) {};  \node (\i b) at (0.6*\i-0.6,-1) {};  }
	\foreach \i in {2,6,10}{
		\node[fill=none] at (0.6*\i-0.6,0) {$\dots$};  \node[fill=none] at (0.6*\i-0.6,-1) {$\dots$};  }
	\draw (1a)--(1b) (3a)--(3b) (4a)to[bend right](8a) (5a)--(5b) (7a)--(7b) (4b)to[bend left](8b) (9a)--(9b) (11a)--(11b) ;
	\node[above=.5ex,fill=none] at (4a.north) {$i$}; \node[above=.5ex,fill=none] at (8a.north) {$j$}; 
	\node[above=.5ex,fill=none] at (1a.north) {1}; \node[above=.5ex,fill=none] at (11a.north) {$D$};
\end{tikzpicture}\;.
\end{equation}

\paragraph{Action on \boldmath$V^{\otimes D}$.}
If $V$ is a real $N$-dimensional vector space with non-degenerate bilinear form $g$, that can be the standard symmetric or symplectic form, one considers integer values of $z=(-1)^{\fb}N$, $N\in\mathbb{N}$ ($\fb=0$ in the symmetric, and $\fb=1$ in the symplectic case). The Brauer algebra $B_D((-1)^{\fb}N)$ acts naturally on tensors of order $D$ that we represent by their components
\begin{equation}
T=T^{a_1a_2\dots a_D}\; e_{a_1}\otimes e_{a_2}\otimes\dots\otimes e_{a_D} \;,
\end{equation}
where $\{e_a\}_{a=1,2,\dots N}$ is a standard basis with respect to the bilinear form $g$ on $V$. An element $\beta\in B_D( (-1)^{\fb}N)$, corresponding to a single Brauer diagram, acts as follows on $T^{a_1a_2\dots a_D}$:
\begin{enumerate}
    \item Place the indices $a_1a_2\dots a_D$ in the top row of the Brauer diagram.
    \item Permute them according to the lines that connect the bottom to the top row.
    \item Contract them with $g$ if they are connected by an arc in the top row.
    \item Add a factor $g^{a_ia_j}$ for each arc in the bottom row.
    \item Multiply the result by 
    $(\eta(\beta))^{\fb}$, where $\eta(\beta)= (-1)^m$ where $m$ is the minimal number of crossings in $\beta$.\footnotemark
\end{enumerate}
\footnotetext{ This sign can be expressed as the sign of oriented pairings in subsection~\ref{sec: translation}. }
Crucially, because of the last point, in applications to $Sp(N)$, $B_D(-N)$ acts in a signed representation.
More explicit, we can associate to $\beta$ a linear map in $\mathrm{End}(V^{\otimes D})$, whose components 
write
\begin{equation}
(\beta)^{a_1a_2\dots a_D}_{b_1b_2\dots b_D}=  \eta(\beta)^{\fb}
    \prod_{\substack{(i,j)\\ i \text{ in the bottom row}\\ \text{connected to }j\text{ in the top row}}} \kern-1em \delta^{a_i}_{b_j}
    \prod_{\substack{(k,l)\\ k \text{ connected to }l\\ \text{by an arc in the bottom row}}} \kern-1em g^{a_ka_l}
    \prod_{\substack{(m,p)\\ m \text{ connected to }p\\ \text{by an arc in the top row}}} \kern-1em g_{b_mb_p} \;,
\end{equation}
and it acts on the tensor components as:
\begin{equation}
    \beta\cdot T^{a_1a_2\dots a_D} = \sum_{b_1,b_2,\dots,b_D} (\beta)^{a_1a_2\dots a_D}_{b_1b_2\dots b_D}\; T^{b_1b_2\dots b_D} \;.
\end{equation}
For example, one has:
\begin{align}
    \sigma_{ij}\cdot T^{a_1\dots a_i\dots a_j\dots a_D} &= T^{a_1\dots a_j\dots a_i\dots a_D} \;, \\
    \beta_{ij}\cdot T^{a_1\dots a_i\dots a_j\dots a_D} &= g^{a_ia_j}\; g_{b_ib_j}\; T^{a_1\dots b_i\dots b_j\dots a_D} \;, \\
    \upsilon\cdot T^{a_1a_2a_3a_4} &=  g^{a_1a_3}\; g_{b_1b_2}\; T^{b_1b_2a_4a_2} \;.
\end{align}
The action is extended to arbitrary elements of the Brauer algebra by linearity.
One can also raise the indices of the linear map using the bilinear form such that:
\begin{equation}
\begin{split}
(\beta)^{a_1a_2\dots a_D, b_{1}\dots b_{D}}&= (\beta)^{a_1a_2\dots a_D}_ {c_1 c_2\dots c_D} \; g^{c_1 b_{1}} \ldots \; g^{c_D b_{D}} 
\\& = \eta(\beta)^{\fb}
    \prod_{\substack{(i,j)\\ i \text{ in the bottom row}\\ \text{connected to }\\ j\text{ in the top row}}} \kern-1em g^{a_i b_j}
    \prod_{\substack{(k,l)\\ k \text{ connected to }l\\ \text{by an arc in the bottom row}}} \kern-1em g^{a_k a_l}
    \prod_{\substack{(m,p)\\ m \text{ connected to }p\\ \text{by an arc in the top row}}} \kern-1em g^{b_m b_p} \;.
\end{split}
\end{equation}
Note that
because of the sign  $\eta(\beta)$ in the definition of the action on $V^{\otimes D}$, the interchange of symmetrization and antisymmetrization when going from $O(N)$ representations to $Sp(N)$ representations is already built in. This can be seen by the fact that in the case where $\sigma$ is a permutation, the sign $\eta(\sigma)$ corresponds to $\sgn(\sigma)$. The components of the linear map associated to $a_\lambda$ and $b_\lambda$ (see \eqref{eq:ab}) thus are:
\begin{equation}
\begin{split}
    (a_\lambda)^{a_1 \ldots a_D}_{b_1 \ldots b_D} &= \sum_{\sigma \in P_\lambda} \sgn(\sigma)^{\fb}  \prod_{\substack{(i,j)\\ j=\sigma(i)}} \delta^{a_i}_{b_j}  \; , \\
    (b_\lambda)^{a_1 \ldots a_D}_{b_1 \ldots b_D}&=\sum_{\tau\in Q_\lambda} \sgn(\tau)^{b+1} \prod_{\substack{(i,j)\\ j=\tau(i)}} \delta^{a_i}_{b_j} \;.
    \label{eq:ab2}
\end{split}
\end{equation}

In conclusion, in the $O(N)$ case ($\fb=0$), $a_\lambda$ now acts as a symmetrizer and $b_\lambda$ as an antisymmetrizer whereas the roles are reversed in the $Sp(N)$ case ($\fb=1$). The product $c_\lambda = a_\lambda \cdot b_\lambda $ thus corresponds to the Young symmetrizer associated to a tableau $\lambda$ when $\fb=0$ and to the symmetrizer associated to the dual tableau $\lambda^\prime$, obtained by permuting the rows and columns of $\lambda$, when $\fb=1$.

\paragraph{Traceless projector.} In order to implement the projection onto irreducible representations of $O(N)$ or $Sp(N)$, the authors of \cite{braueralg} build a universal traceless projector, which we 
{introduce} here, for the sake of completeness. The main building block of this projector is:
\begin{equation}
A_D=\sum_{1\leq i<j\leq D} \beta_{ij} \;\in B_D((-1)^{\fb}N) \;.
\end{equation}
Let us now list some important properties of $A_D$:
\begin{itemize}
    \item It commutes with all elements of $\mathbb{C}\fS_D\subset B_D(N)$. Thus, in particular, it commutes with Young symmetrizers.
    \item The action of $A_D$ on $V^{\otimes D}$ is diagonalizable.
    \item The kernel $\ker A_D\subset V^{\otimes D}$ is exactly the space of traceless tensors.
    \item Its non-zero eigenvalues are in $(-1)^{\fb}\mathbb{N}$.
\end{itemize}
The proof of these statements can be found in \cite{braueralg}, and the universal traceless projector is 
given by:
\begin{equation}\label{eq: tracelessproj}
    \mathfrak{P}_D= \sum_{\alpha\text{ non-zero eigenvalue of }A_D} \big(1-\frac{1}{\alpha}A_D \big) \;.
\end{equation}
Explicit formulas for the non-zero eigenvalues $\alpha$ are also given in \cite{braueralg}.

\subsection{Sign of directed pairings}\label{sec: pairings}

In this subsection, we define the sign given by two oriented pairings and give some of its properties.

Consider two oriented pairings $\vM_1$ and $\vM_2$ on a set of $2D$ elements, suppose these two pairings are given by
\begin{equation}
\begin{split}
    &\vM_1 = \{ (i_1,i_2), \ldots , (i_{2D-1},i_{2D}) \} \,, \\
    &\vM_2 = \{ (j_1,j_2), \ldots , (j_{2D-1},j_{2D}) \} \,.
\end{split}
\end{equation}
The sign $\epsilon(\vM_1,\vM_2)$ of the two pairings is defined as the sign of the permutation $\sigma=\big(\begin{smallmatrix}i_1&i_2&\dots&i_{2D-1}&i_{2D}\\ j_1&j_2&\dots&j_{2D-1}&j_{2D}\end{smallmatrix}\big)$:
\begin{equation}
    \epsilon(\vM_1,\vM_2) = \sgn \left(\big(\begin{smallmatrix}i_1&i_2&\dots&i_{2D-1}&i_{2D}\\ j_1&j_2&\dots&j_{2D-1}&j_{2D}\end{smallmatrix}\big) \right)
    \label{eq: sign}
\end{equation}
We give here a list of some of the properties of the sign $\epsilon(\vM_1,\vM_2)$:
\begin{enumerate}
     \item It is symmetric under permutation of its arguments: 
    \begin{equation}
      \epsilon(\vM_1,\vM_2) = \epsilon(\vM_2,\vM_1)\; .
       \label{eq:prop1}
    \end{equation}
    \item For three pairings $\vM_1$, $\vM_2$, $\vM_3$ on the same set, one has:
        \begin{equation}
            \epsilon(\vM_1,\vM_2) = \epsilon(\vM_1,\vM_3) \epsilon(\vM_2,\vM_3).
           \label{eq:prop2}
        \end{equation}
    \item For two pairings $\vM_1$, $\vM_2$ on a first set $\mathcal{S}_1$ of $2D$ elements and two pairings $\vM_3$, $\vM_4$ on a second set $\mathcal{S}_2$ of $2p$ elements, 
    one has:
    \begin{equation}
       \epsilon(\vM_1,\vM_2) \epsilon(\vM_3,\vM_4) = \epsilon(\vM_1 \sqcup \vM_3,\vM_2 \sqcup \vM_4)\;.
      \label{eq:prop3}
    \end{equation}
    \item Consider a set of elements $\mathcal{S}_v$ and two pairings $\vM_1$ and $\vM_2$ on this set. Depict each elements of $\mathcal{S}_v$ as a node and each pair in $\vM_1$ and $\vM_2$ as an oriented edge pointing from the first element to the second, and of color $1$ for the pairs in $\vM_1$ and color $2$ for the ones in $\vM_2$. The sign $\epsilon( \vM_1,\vM_2 )$ can be written as:
    \begin{equation}
        \epsilon( \vM_1,\vM_2 ) = (-1)^{F_{1/2,even}} \, .
        \label{eq:prop4}
    \end{equation}
    In the equation above, $F_{1/2,even}$ is the number of even faces of color $1$ and $2$ of the graphical representation described above. An even, resp.~odd, face of color $1$ and $2$ is defined as a closed cycle of alternating colors $1$ and $2$ where an even, resp.~odd, number of edges point in one direction around the cycle. Because each face consists of an even number of edges, this notion is well defined. 
\end{enumerate}

In the sequel, the sign
$\epsilon$ plays a crucial role in the proof of our main theorem. Let us first link this quantity to the Brauer algebra and connect them to the tensor models. 

\subsection{Pairings, the Brauer algebra, propagators and projectors}
\label{sec: translation}

In this subsection we exhibit the connection between the notions of subsections~\ref{sec: irreps} and \ref{sec: pairings}, and
propagators in the random tensor models we study in this paper.

The relation between Brauer diagrams and pairings is straightforward, as each Brauer diagram is a pairing of $2D$ vertices. 
Moreover, the sign $\eta(\beta)$ that appears in the description of the action of $B_D((-1)^{\fb}N)$ on $V^{\otimes D}$, can be expressed as the sign of two directed pairings by the following construction:
\begin{enumerate}
    \item Label the vertices in the Brauer diagram $1,2,\dots,D$ in the top row and $D+1, D+2,\dots, 2D$ in the bottom row.
    \item Let $\vec{\beta}$ be the directed pairing induced by $\beta$, where edges are oriented from top to bottom, left to right in the top row and right to left in the bottom row.
    \item Let $\vM_{ref} =\{(1,D+1),(2,D+2),\dots (D,2D)\}$ be the reference pairing, that pair top to bottom vertices.
    \item One then has:  $\eta(\beta)=\epsilon(\vec{\beta},\vM_{ref})\;$.
\end{enumerate}
This follows from the use of \eqref{eq:prop4}.
Moreover, $\beta^{a_1 \ldots a_D, a_{D+1} \ldots a_{2D}}$ admits a compact form in term of the oriented pairing $\vec{\beta}$:
\begin{equation}
\begin{split}
(\beta)^{a_1a_2\dots a_D, a_{D+1}\dots a_{2D}}= \epsilon(\vec{\beta},\vM_{ref})^{\fb}
    \prod_{(i,j) \in \vec{\beta}} g^{a_i a_j} \;.
\end{split}
\label{eq:pairing_brauer}
\end{equation}

In a random tensor model, with tensors of order $D$, living in a representation $R\subset V^{\otimes D}$ of the group $O(N)$ or $Sp(N)$, a propagator is a $O(N)$- or $Sp(N)$-linear map $C\in \mathrm{End}(R)$. As the Brauer algebra is isomorphic (for $N$ large enough) to this space of $O(N)$- or $Sp(N)$-linear maps, each propagator is also an element of $B_D((-1)^{\fb}N)$.

As a consequence of Schur's lemma, if the representation $R$ is irreducible, $C$ is proportional to the identity on $R$, and if $R$ is reducible and decomposes into a direct sum of distinct irreducible representations 
  $R_i$
($R=\bigoplus_{i=1}^k R_i$),
then $C$ decomposes as well into a direct sum of maps $P_i$, each proportional to the identity on $R_i$. 

Denoting by ${P_R}\in\mathrm{End}(V^{\otimes D})$ the orthogonal projector on $R$, i.e.~$\mathrm{im}({P_R})=R$. The propagator can be trivially extended to the whole space $V^{\otimes D}$ by ${C}\circ {P_R}$. Thus, reformulating the implications of Schur's Lemma: If $R$ is irreducible the propagator is proportional to the projector on $R$, and if $R$ decomposes into distinct irreducible representations as above, the propagator is a linear combination of the projectors on the $R_i$.

When studying tensor models from a quantum field theoretical perspective, one is interested in the calculation of expectation values of the form:
\begin{equation}
    \langle f(T) \rangle =  \frac{\Big[e^{\partial_T \boldsymbol{C}\partial_T} e^{-V(T)} f(T) \Big]_{T=0}}{\Big[e^{\partial_T \boldsymbol{C}\partial_T} e^{-V(T)} \Big]_{T=0}} 
    \;,
\end{equation}
where $V(T)$ and $f(T)$ are invariant under the group action, and $\partial_T \boldsymbol{C}\partial_T$ is a short hand notation for the Laplacian-like second order differential operator:
\begin{equation}
  \partial_T \boldsymbol{C} \partial_T := \sum_{a^1_1,\dots,a^1_D,a^2_1,\dots,a^2_D=1}^N \frac{\partial}{\partial T^{a^1_1\dots a^1_D}} C^{a^1_1\dots a^1_D,\, a^2_1\dots a^2_D}  \frac{\partial}{\partial T^{a^2_1\dots a^2_D}} \;.
\end{equation}
Note that indices are raised and lowered by the non-degenerate bilinear form, as usual.
In the above formulation, the tensors are elements of $R$, i.e.~have some non-trivial symmetry. But one can as well consider every tensor $T\in R$ to arise from the projection of a tensor $\tilde{T}\in V^{\otimes D}$ without symmetry under permutation of its indices. Thus, if we supplement the derivative operator with the appropriate projector, only modes obeying the symmetry (tensors in $R$) propagate and $T$ can be replaced by $\tilde{T}$:
\begin{equation}
\label{eticheta}
    \langle f(T) \rangle =  \frac{\Big[e^{\partial_{\tilde{T}} (\boldsymbol{C P_R})\partial_{\tilde{T}}} e^{-V(\tilde{T})} f(\tilde{T}) \Big]_{\tilde{T}=0}}{\Big[e^{\partial_{\tilde{T}} (\boldsymbol{C P_R})\partial_{\tilde{T}}} e^{-V(\tilde{T})} \Big]_{\tilde{T}=0}} 
    \;,
\end{equation}
with the convention $\partial_{\tilde{T}}\tilde{T}=id_{V^{\otimes D}}$.

\section{The graded tensor model}\label{sec: model}

Let $T^{a_1 \ldots a_D}$ be the components of a generic random tensor  with $D$ indices (an order $D$ tensor). Each index of the tensor ranges from $1$ to $N$, the tensor has thus $N^D$ independent components. 
As already mentioned above, we introduce a parameter $\fb$, equal to $0$ or $1$, that defines the symmetry properties of the tensor. If $\fb=0$, resp.~$\fb=1$, the tensor transforms in some representation $R$ of order $D$ of the orthogonal group $O(N)$, resp.~symplectic group $Sp(N)$. Using Einstein summation convention the group action writes:
\begin{equation}
T^{a_1 \ldots a_D} \to T'^{a_1 \ldots a_D} =  \tensor{(O_{\fb})}{^{a_1}_{b_1}}  \tensor{(O_{\fb})}{^{a_2}_{b_2}}\ldots \tensor{(O_{\fb})}{^{a_D}_{b_D}} T^{b_1 \ldots b_D} \, , \quad O_\fb\in {\begin{cases}
O(N) &, \fb=0 \\
Sp(N) &, \fb=1
\end{cases}} .
\end{equation}
Moreover, the indices of the tensor are contracted using a graded symmetric form $g^{\fb}$ such that $g^{\fb}_{a b} = (-1)^{\fb} g^{\fb}_{b a}$. One has:
\begin{equation}
g^{\fb}_{a b} = \smash[b]{\begin{cases}
\delta_{a b}&, \fb=0 \\
\omega_{a b}&, \fb=1
\end{cases}} ,  \quad \text{with} \quad 
\delta = \left( 
    \begin{array}{c|c} 
      \mathds{1}_{N/2}  & 0 \\ 
      \hline 
      0 & \mathds{1}_{N/2}  
    \end{array} 
    \right)  
\quad \text{and} \quad 
\omega =     \left( 
    \begin{array}{c|c} 
      0 & \mathds{1}_{N/2}   \\ 
      \hline 
      -\mathds{1}_{N/2} & 0  
    \end{array} 
    \right) \; .   
\end{equation}

Thus, the tensor components are fermionic (odd graßmannian) if $\fb=1$ and $D$ odd, and bosonic otherwise
(the parity of the tensor components is $\fb D \mod 2$). 

\paragraph{Invariants and directed stranded graphs.} By contracting indices with $g^{\fb}$ one can build invariant polynomials in the tensor components. 
Unlike the graded colored tensor models studied previously in \cite{Duality_Hannes, kepplermuller23}, two indices at different positions can now be contracted. Therefore, the invariants do not admit a graphic representation in term of directed edge colored graphs but they do admit one in terms of directed stranded graphs such that:
\begin{itemize}
\item each tensor is represented by a set of $D$ nodes labeled by its indices.
\item each contraction of indices is represented by a strand connecting the corresponding nodes.
\end{itemize}
\begin{definition}[Stranded Graph] We encode a directed stranded graph $\vec{\mathcal{S}}$ with $D$ strands by a set of nodes $V(\vec{\mathcal{S}})$ with $|V(\vec{\mathcal{S}})|$ elements, that come in groups of $D$,  and a set of edges, called strands, $\vec{E}(\vec{\mathcal{S}})$, such that $\vec{E}(\vec{\mathcal{S}})$ is a directed pairing of $V(\vec{\mathcal{S}})$.
One often refers to the 
$D$ nodes as vertices. If two such vertices are directly connected by $D$ strands one often refers to this collection of $D$ strands as edge. 

We also denote the undirected version of a directed stranded graph by $\mathcal{S}$. Two examples are drawn in Figure~\ref{fig:stranded}.
\end{definition}
\begin{figure}
    \centering
    \includegraphics[scale=0.8]{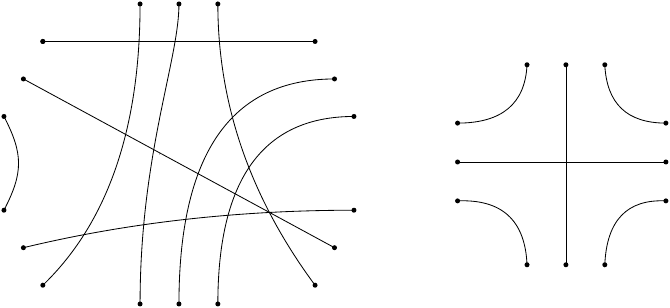}
    \caption{Two stranded graphs for $D=3$. }
    \label{fig:stranded}
\end{figure}
As a shorthand notation we write $a_\cD =(a_1, a_2, \dots , a_D)$ for the sequence of $D$ indices. 
To each directed stranded graph $\vec{\mathcal{S}}$ is then associated an invariant whose expression reads:
\begin{equation}
I_{\vec{\mathcal{S}}}(T) = \Big( \prod_{(i,j) \in \vec{M}_{ref}} T^{a^i_\mathcal{D}} T^{a^j_\mathcal{D}} \Big) \; \epsilon (\vec{M}^D_{ref},\vec{E}(\vec{\mathcal{S}}))^{\fb} \prod_{(k,l) \in \vec{E}(\vec{\mathcal{S}})} g^{\fb}_{\, k  l} \; .
\label{eq: invariant}
\end{equation}
In the equation above, $\vec{M}_{ref}$ is an arbitrary reference pairing of $2p=|V(\vec{\mathcal{S}})|/D$ tensors. The pairing $\vec{M}^D_{ref}$ is a directed pairing of the indices of the tensors given by the disjoint union of $D$ copies of $\vec{M}_{ref}$. An illustration is given 
in Figure~\ref{fig:M_0D}. The term $\epsilon (\vec{M}^D_{ref},\vec{E}(\vec{\mathcal{S}}))^{\fb}$ is the sign of the pairing $\vec{M}^D_{ref}$ with respect to $\vec{E}(\vec{\mathcal{S}})$;
this sign is defined in \eqref{eq: sign}. 

Introducing the sign of the pairings in the expression of an invariant fixes the ambiguity induced by the graded symmetry of $g^{\fb}$. 
Two invariants associated to two directed version of the same stranded graph $\mathcal{S}$ are equal, $I_{\vec{\mathcal{S}}}(T)$ is then a class function and we can choose a single representative of $\mathcal{S}$ in the action of our model, more comments on this can be found in \cite{kepplermuller23}. As a consequence we drop the arrow in the notation if we refer to the undirected version of the graph, and if the quantity does not depend on the chosen orientation of the graph.

As one may contract indices of different positions together, there are several possible quadratic invariants. We group them into a quadratic term of the form
$T^{a_\mathcal{D}} C^{-1}_{a_\mathcal{D} b_\mathcal{D}} T^{b_\mathcal{D}}$. 
The propagator of the model is given by
\begin{equation}
C^{a_\mathcal{D} b_\mathcal{D}} = \sum_{M \in  \mathbf{M}\{a_\cD b_\cD \}} \gamma_{M}\; \epsilon(\vec{M}, \vec{M}_{ref,C})^{\fb} \prod_{(i,j) \in \vec{M}} g_{\fb}^{\, ij}\;, \qquad \gamma_M\in\mathbb{R} \; ,
\label{eq: propagator}
\end{equation}  
where 
$g^{\, a_1 b_1}_{\fb}$ denotes the components of the inverse of $g^{\fb}$ such that $g_{\, a c}^{\fb} g^{\, c d}_{\fb} = \delta_{a}^{\;d}$. Moreover, 
$\mathbf{M}\{a_\cD b_\cD \}$ is the set of non oriented pairing on the set of $2D$ indices $a_{\cD} \cup b_{\cD}$ and $\vec{M}$ is a chosen oriented version of $M$.  The pairing $\vec{M}_{ref,C}$ is a reference pairing of the indices given by:
\begin{equation}
    \vec{M}_{ref,C} = \{ (a_1,b_1), \ldots, (a_D, b_D) \} \; .
\end{equation} 
This corresponds to the case where each index of the first tensor propagates to the index at the same position in the second tensor (see Figure~\ref{fig:M_ref}).  
\begin{figure}
    \centering
    \includegraphics[scale=1.0]{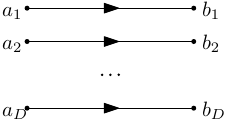}
    \caption{Graphical representation of $\vM_{ref,C}$}
    \label{fig:M_ref}
\end{figure}

Let us emphasize that the product $\boldsymbol{C P_R}$ in \eqref{eticheta} is a particular case of the general propagator \eqref{eq: propagator}, when ${\boldsymbol{C}}= \mathds{1}$ and $ \boldsymbol{P_R}$ is the projector on the irreducible representation $R$ of $O(N)$ or $Sp(N)$. This is explained in detail in Appendix~\ref{apendice}.

As noted in subsection~\ref{sec: translation}, $C^{a_\mathcal{D} b_\mathcal{D}}$ is an element of the Brauer algebra $B_D((-1)^{\fb}N)$. Each pairing $M$ in the sum represents a Brauer diagram and the factors $\gamma_M$ are the coefficients in the linear combination. The reference pairing $\vec{M}_{ref,C}$ coincides with $\vec{M}_{ref}$ from subsection~\ref{sec: pairings}. Note that Brauer diagrams are conventionally read from top to bottom, whereas propagators are usually drawn from left to right.

\begin{definition}[Graded Tensor Model with Symmetry]\label{def: model}
We define the \emph{graded tensor model with symmetry} $R$ by the measure:
\begin{equation} \label{eq: model-def}
\begin{aligned}
d\mu[T]&=e^{-S[T]}\ [dT], \quad [dT]=\zeta\prod_{a_{\mathcal{D}}} dT^{a_1\dots a_D}\; ,
\\ \text{with}\quad
S[T] &= T^{a_\mathcal{D}} C^{-1}_{a_\mathcal{D} b_\mathcal{D}} T^{b_\mathcal{D}} + \sum_{\substack{\mathcal{S}\ \text{connected,} \\ |V(\mathcal{S})|>2}} \frac{\lambda_{\mathcal{S}}}{|V(\mathcal{S})|/D}\, I_{\mathcal{S}}(T) \;, 
\end{aligned}
\end{equation}
and normalization $\zeta$ such that $\int d\mu[T]=1$ for $\lambda_\cS = 0$ $\forall \lambda_\cS$. All tensors are 
elements of the $O(N)$ (for $\fb=0$), resp.~$Sp(N)$ (for $\fb=1$), representation $R$.
\end{definition}

In the definition above, the constant $\lambda_{\mathcal{S}}$ is the coupling constant of the invariant associated to $\mathcal{S}$. The partition function of this models writes:
\begin{equation}
    Z= \int d\mu[T] = \left[ e^{\partial_T \boldsymbol{C} \partial_T}  e^{\sum \frac{\lambda_{\mathcal{S}}}{|V(\mathcal{S})|/D} I_{\cS(T)}}\right]_{T=0} \, ,
\end{equation}
where the derivative representation \cite{Brydges:2014,salmhofer,gurau} of the Gaussian integral is used and $\partial_T \boldsymbol{C} \partial_T$ is a short-hand notation for:
\begin{equation}
  \partial_T \boldsymbol{C} \partial_T := \frac{\partial}{\partial T^{a^1_\cD}} C^{a^1_\cD a^2_\cD}  \frac{\partial}{\partial T^{a^2_\cD}} \,.
\end{equation}
When making use of the derivative representation, as discussed in Section~\ref{sec: translation}, we can take the tensors to have no symmetries under permutations of their indices, but instead incorporate an appropriate projector on the space $R$ in the definition of the propagator.

\section{Proof of the main result}
\label{sec: proof}

In this section, we prove the main theorem of our paper. We show that the partition function of the graded tensor model is invariant under the change of parameters $\fb \to \fb+1 \mod 2$ and $N \to -N$. By choosing the propagator according to a given symmetry specified by the $O(N)$ or $Sp(N)$ representation $R$ this implies the stated duality. 
The appropriate choice of the propagator as an element of the respective Brauer algebra was discussed in Section~\ref{sec: pairings}.
From a mathematical point of view, the choice is implemented by fixing the pairing $\vM$ and constants $\gamma_M$ in \eqref{eq: propagator} accordingly.

Let us first recall the commutation relation of the tensor components:
 $T^{a_\mathcal{D}} T^{b_\mathcal{D}} = (-1)^{\fb D }\; T^{b_\mathcal{D}} T^{a_\mathcal{D}}$.
The Gaußian (free) expectation value $\langle T^{a^1_\mathcal{D}} \ldots T^{a^{2p}_\mathcal{D}}\rangle_0$ of $2p$ tensors whose order is encoded by $\vM_{ref}$ is defined as:
\begin{equation}
\langle T^{a^1_\cD} \ldots T^{a^{2p}_\cD} \rangle_{0}
= \left[ e^{\partial_T \boldsymbol{C} \partial_T} \;  T^{a^1_\cD} \dots T^{a^{2p}_\cD} \right]_{T=0}\;.
\end{equation}
 For our model, Wick's theorem 
expresses this expectation as a sum over pairings of $2p$ elements:
\begin{equation}
\langle T^{a^1_\mathcal{D}} \ldots T^{a^{2p}_\mathcal{D}}\rangle_0 = \sum_{M_0 \in \mathbf{M}_{2p}} \epsilon( \vM_{ref},\vM_0)^{\fb D} \Big( \prod_{(i,j) \in \vM_0} C^{a^i_\cD a^j_\cD} \Big) \;.
\label{eq:wick}
\end{equation}
The sign $\epsilon( \vM_{ref,2p},\vM_0)^{\fb D}$ in \eqref{eq:wick} takes into account the type (bosonic/fermionic) of the tensor components. The directed pairing $\vM_0$ is an arbitrary oriented version of $M_0$, but notice that the term $ \epsilon( \vM_{ref,2p},\vM_0)^{\fb D} \Big( \prod_{(i,j) \in \vM_0} C^{a^i_\cD a^j_\cD} \Big)$ is invariant under reorientation of pairs in $\vM_0$.

The Gaußian (free) expectation of an invariant $I_\cS(T)$ of order $2p$ specified by a stranded graph $\cS$ is defined as:
\begin{equation}
\label{free}
    \langle I_{\cS}(T) \rangle_0 = \left[ e^{\partial_T \boldsymbol{C} \partial_T}\;  I_\cS(T) \right]_{T=0}\;.
\end{equation}
This expectation value can be computed by pairing the $2p$ groups of $D$ vertices in $\cS$ by propagators \eqref{eq: propagator}. We represent this pairing by edges of a new color $0$, each consists again of $D$ strands.

The result is a sum over $2$-colored stranded graphs $\cG$, such that $\cS\subset\cG$ is the maximal subgraph of color $1$ (see Fig.~\ref{fig:2colorstranded} for an example of such a graph).

\begin{lem}\label{lem}
The Gaussian expectation \eqref{free} writes:
\begin{equation}
\big\langle I_{\cS} (T)\big\rangle_0 = \sum_{\substack{\cG,\; \cS \subset \cG \\ |V(\cG)|=2pD}} \gamma_{\cG} \left((-1)^{\fb} N \right)^{F(\cG)} \;,
\label{eq: lemma}
\end{equation}
where the power of $N$ is given by the number of faces of $\cG$. Moreover, the factor $\gamma_\cG$ is a product of weights associated to the edges of color $0$, given by the expression of the propagator in \eqref{eq: propagator}. It writes:
\begin{equation}
    \gamma_\cG = \prod_{e\in E_0(\cG)} \gamma_{M^e} \;,
\end{equation}
with $E_0(\cG)$ the set of edges of color 0 and $M^e$ the pairing (Brauer diagram) defining the path of the $D$ strands of the color $0$ edge $e$.
\end{lem}
\begin{figure}
    \centering
    \includegraphics[scale=0.7]{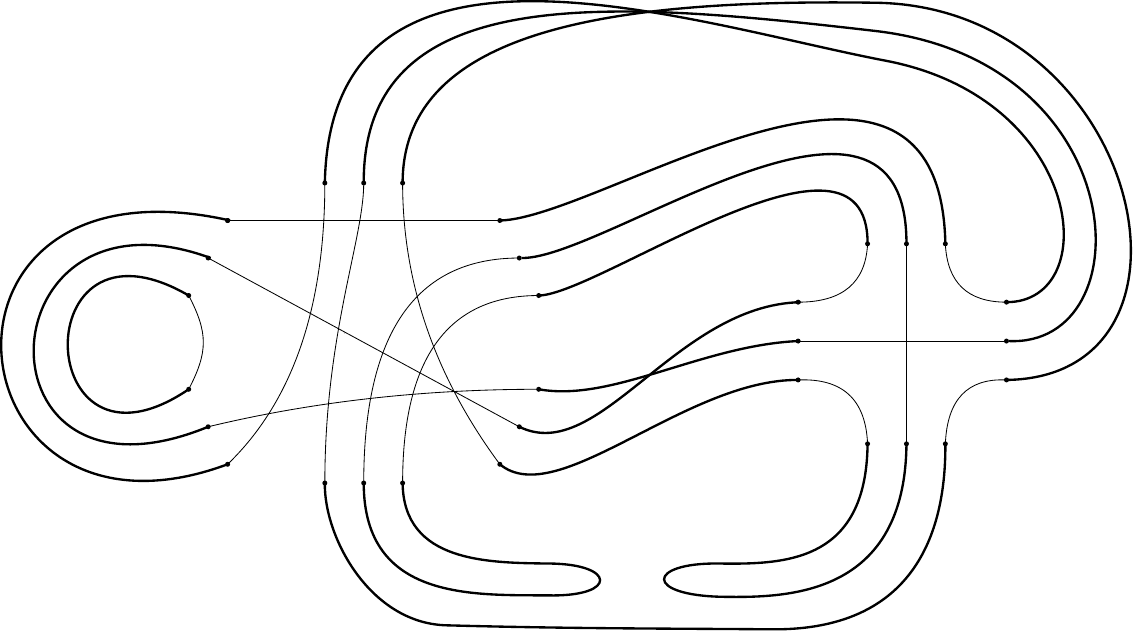}
    \caption{A 2-colored stranded graph for $D=3$, obtained by connecting the two stranded graphs in Figure~\ref{fig:stranded} by propagator edges (elongated Brauer diagrams). The color 0 is represented by thick lines and color 1 by thin lines.}
    \label{fig:2colorstranded}
\end{figure}
\begin{proof}
Applying Wick's theorem \eqref{eq:wick} to the formula of an invariant of order $2p$ specified by a directed stranded graph $\vec{\cS}$ \eqref{eq: invariant}, leads to the following form of the Gaußian expectation:  
\begin{equation}
    \begin{split}
    \langle I_{\cS} \rangle_0 &= \big\langle \prod_{(i,j) \in \vec{M}_{ref}} T^{a^i_\mathcal{D}} T^{a^j_\mathcal{D}} \big\rangle_0 \epsilon (\vec{M}^D_{ref},\vec{E}(\vec{\mathcal{S}}))^{\fb} \left(\prod_{(k,l) \in \vec{E}(\vec{\mathcal{S}})} g^{\fb}_{\, k  l} \right)\, \\
    &= \sum_{M_0 \in \mathbf{M}_{2p}} \epsilon (\vec{M}^D_{ref},\vec{E}(\vec{\mathcal{S}}))^{\fb} \epsilon( \vM_{ref},\vM_0)^{\fb D} \Big( \prod_{(i,j) \in \vM_0} C^{a^i_\cD a^j_\cD} \Big) \left(\prod_{(k,l) \in \vec{E}(\vec{\mathcal{S}})} g^{\fb}_{\, k  l} \right) \, .
    \end{split}
    \label{eq:Gauss_inv}
\end{equation}
First, the dependence on the reference pairing can be eliminated using the properties \eqref{eq:prop3} and \eqref{eq:prop2} of the sign $\epsilon$ such that:
\begin{equation}
    \epsilon (\vec{M}^D_{ref},\vec{E}(\vec{\mathcal{S}}))^{\fb} \epsilon( \vM_{ref},\vM_0)^{\fb D} =  \epsilon(\vM_0^D, \vec{E}(\vec{\mathcal{S}}))^{\fb} \, ,
    \label{eq:sign_reform}
\end{equation}
where $\vM_0^D$ is the oriented pairing given by the disjoint union of $D$ copies of $\vM_0$. This pairing can be seen as taking each pairs of tensors in $\vM_0$ and pairing their indices, respecting their position. An illustration can be found in Figure~\ref{fig:M_0D}.

\begin{figure}
    \centering
    \includegraphics[scale=1.0]{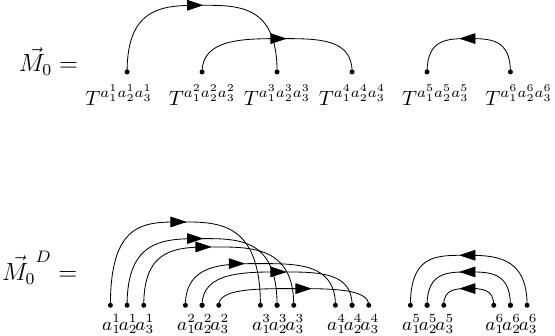}
    \caption{The pairing $\vM_0$ of tensors is promoted to the pairing $\vM_0^D$ of their indices.}
   \label{fig:M_0D}
\end{figure}

Second, we rewrite the term $\Big( \prod_{(i,j) \in \vM_0} C^{a^i_\cD a^j_\cD} \Big)$ using \eqref{eq: propagator} as:
\begin{equation}
    \begin{split}
     \prod_{(i,j) \in \vM_0} C^{a^i_\cD a^j_\cD} &= \prod_{(i,j) \in \vM_0} \left( \sum_{M_{ij} \in  \mathbf{M}\{a^{i}_\cD a^j_\cD \}} \gamma_{M} \epsilon(\vec{M}_{ij}, \vec{M}_{ref,ij,C})^{\fb} \prod_{(m,n) \in \vec{M}_{ij}} g_{\fb}^{\, m n} \right) \\
     &= \sum_{M_{tot} \in \mathbf{M}_{tot}} \gamma_{M_{tot}} \epsilon(\vM_{tot},\vM_0^D)^{\fb} \left(  \prod_{(m,n)\in \vM_{tot}} g_{\fb}^{\, m n} \right) \;,
    \end{split}
    \label{eq:Prop_reform}
\end{equation}
where $\mathbf{M}\{a^{i}_\cD a^j_\cD \}$ is the set of pairings of elements $a^{i}_\cD\cup a^j_\cD $, and $\mathbf{M}_{tot}$ is the set of pairings given by the disjoint union of all $\mathbf{M}\{a^{i}_\cD a^j_\cD \}$ with $(i,j) \in \vM_0$.
\begin{figure}
    \centering
    \includegraphics[scale=1.0]{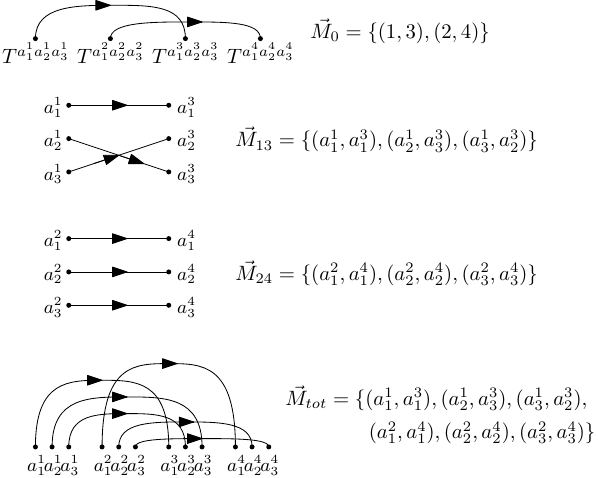}
    \caption{The construction of a pairing $\vM_{tot}$, from a pairing $\vM_0$ of four tensors of rank 3, as well as two pairings of indices $\vM_{13}$ and $\vM_{24}$. In spirit, these describe propagators connecting the different tensors.}
    \label{fig:M_tot}
\end{figure}
A pairing $M_{tot}\in\mathbf{M}_{tot}$ is therefore the disjoint union of $p$ pairings belonging to the sets $\mathbf{M}\{a^{i}_\cD a^j_\cD \}$. An example is shown in Figure~\ref{fig:M_tot}. Denoting these $p$ pairings as $M^1  \ldots M^p$, the factor $\gamma_{M_{tot}}$ is equal to:
\begin{equation}
    \gamma_{M_{tot}}= \prod_{x=1}^{p} \gamma_{M^{x}} \, .
\end{equation}
We used here the fact that, by construction, the disjoint union of the $M_{ref,ij,C}$ is equal to $\vM_0^D$. 
This comes from the fact that both contract the indices of a pair of tensors present in $\vM_0$, respecting the position of indices (see Figure~\ref{fig:M_refP}).

\begin{figure}
    \centering
    \includegraphics[scale=1.0]{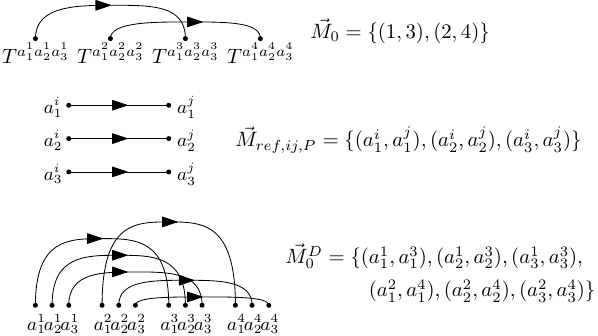}
    \caption{Illustration of the correspondence between $\vM_0^D$ and $\vM_{ref,ij,C}$.}
    \label{fig:M_refP}
\end{figure}

Inserting \eqref{eq:sign_reform} and \eqref{eq:Prop_reform} in \eqref{eq:Gauss_inv}, we obtain:
\begin{equation}
    \begin{split}
    \langle I_{\vec{\mathcal{S}}} \rangle_0 &= \sum_{\substack{ M_0 \in \mathbf{M}_{2p} \\ M_{tot} \in \mathbf{M}_{tot}}}  \gamma_{M_{tot}} \epsilon(\vM_0^D, \vec{E}(\vec{\mathcal{S}}))^{\fb} \epsilon(\vM_{tot},\vM_0^D)^{\fb} \left(  \prod_{(m,n)\in \vM_{tot}} g_{\fb}^{\, m n} \right) \left(\prod_{(k,l) \in \vec{E}(\vec{\mathcal{S}})} g^{\fb}_{\, k  l} \right) \; .
    \end{split}
    \label{eq}
\end{equation}
This further leads to using property \eqref{eq:prop2}:
\begin{equation}
   \langle I_{\cS} \rangle_0 = \sum_{\substack{ M_0 \in \mathbf{M}_{2p} \\ M_{tot} \in \mathbf{M}_{tot}}}  \gamma_{M_{tot}} \epsilon(\vM_{tot}, \vec{E}(\vec{\mathcal{S}}))^{\fb} \left(  \prod_{(m,n)\in \vM_{tot}} g_{\fb}^{\, m n} \right) \left(\prod_{(k,l) \in \vec{E}(\vec{\mathcal{S}})} g^{\fb}_{\, k  l} \right) \, . 
\end{equation}

Adding oriented edges of a new color $0$ to $\mathcal{S}$, according to $\vM_{tot}$, yields a $2$-color directed stranded graph $\vcG$. We define a face in $\cG$ as a cycle with strands of alternating colors. Along a face, $g_{\fb}$ and its inverse alternate and all indices are summed. Therefore, each face contributes a factor $N$. However, because of the graded symmetry $g^{\fb}_{\, a c}= (-1)^{\fb} g^{b}_{\, ca}$ a face also picks up a factor $(-1)^{\fb}$ if an odd number of strands point in one of the two directions around the face, we characterize such a face to be \textit{odd}, otherwise a face is called \textit{even}. The term $\left(  \prod_{(m,n)\in \vM_{tot}} g_{\fb}^{\, m n} \right) \left(\prod_{(k,l) \in \vec{E}(\vec{\mathcal{S}})} g^{\fb}_{\, k  l} \right)$ thus contributes:
\begin{equation}
  \left(  \prod_{(m,n)\in \vM_{tot}} g_{\fb}^{\, m n} \right) \left(\prod_{(k,l) \in \vec{E}(\vec{\mathcal{S}})} g^{\fb}_{\, k  l} \right) = (-1)^{\fb F_{odd}(\vcG)} N^{F(\cG)}  
\end{equation}
Using property \eqref{eq:prop4} we also rewrite the term $\epsilon(\vM_{tot}, \vec{E}(\vec{\mathcal{S}}))^{\fb}$ as:
\begin{equation}
    \epsilon(\vM_{tot}, \vec{E}(\vec{\mathcal{S}}))^{\fb} = (-1)^{F_{even}(\vcG) b} \, ,
\end{equation}
where $F_{even}(\vcG)$, resp.~$F_{odd}(\vcG)$, denotes the number of even, resp.~odd, faces of $\vcG$ and $F(\cG) = F_{odd}(\vcG) + F_{even}(\vcG)$ is the total number of faces of $\cG$, which does not depend on any chosen orientation.

The expectation value $\langle I_{\cS} \rangle_0$ can thus be evaluated as a sum over $2$-colored stranded graphs $\cG$:
\begin{equation}
    \langle I_{\cS} \rangle_0 = \sum_{\substack{\cG,\; \cS \subset \cG \\ |V(\cG)|=2pD} } \gamma_{M_{tot}} \left((-1)^{\fb} N \right)^{F(\cG)} \;.
\end{equation}
This concludes the proof.
\end{proof}

Each term in \eqref{eq: lemma} is invariant under the transformation:
\begin{equation}
    \label{dualitate}
    \fb \to \fb+1 \mbox{ and } N \to -N\; .
\end{equation}
Thus, this transformation does not affect the Gaußian expectation value of any invariant nor the amplitude of its graphs and is hence a duality of our model.

The invariance of the partition function under the duality follows directly from the above statement, using a perturbative expansion of the interaction part of the action:
\begin{equation}
    \begin{split}
    Z &=  \left[e^{\partial_T \boldsymbol{C} \partial_T} \sum_{\{p_\cS \geq 0\}} \prod_{\cS} \frac{1}{p_\cS!} \Big( \frac{\lambda_\cS}{|V(\cS)|/D} I_{\cS}(T)  \Big)^{p_\cS}\right]_{T=0} \\
    &=  \sum_{\{p_\cS \geq 0\}} \prod_{\cS} \frac{1}{p_\cS!} \Big( \frac{\lambda_\cS}{|V(\cS)|/D} \Big)^{p_\cS} \big\langle \prod_{\cS} I_{\cS}(T)^{p_\cS}  \big\rangle_0 \, .
    \end{split}
\end{equation}
Since any product of invariants is a single disconnected invariant, the factor $\big\langle \prod_{\cS} I_{\cS}(T)^{p_\cS}  \big\rangle_0$ is invariant under the duality \eqref{dualitate}. Hence the partition function of the model is invariant under \eqref{dualitate}.

As usually, expectation values of invariants are calculated by taking derivatives of $\ln Z$ with respect to the couplings $\lambda_\cS$ (see, for example, \cite{gurau, kepplermuller23}). Diagrammatically, the derivative marks a $1$-colored stranded subgraph of type $\cS$ and this leads to the conclusion stating that the expectation value of an invariant can be experesssed a s a formal sum over $2-$colored stranded graphs (see again \cite{kepplermuller23}).

\section{Illustration: totally symmetric and antisymmetric tensor models}\label{sec: example}

In this section, we exhibit the general duality result proved in the previous section for the particular case of totally symmetric and antisymmetric tensor models.

\subsection{$O(N)$ tensor models}

The vector space $V$ is, in this case, an ordinary even (bosonic) $N$-dimensional real vector space and the tensor product space $V^{\otimes D}$ is an even vector space $\forall D\in\mathbb{N}$. 
As already explained above, the grading parameter now takes the value $\fb=0$.

To the $GL(N)$ representation of totally symmetric tensors $\mathrm{Sym}^D(V)$ of order $D$ is associated the following Young diagram:
\begin{equation}\lambda_{S}=\underbrace{
\begin{ytableau}
\ & \ & & \none[\dots] & \ 
\end{ytableau}}_{\text{length}\ D} \;.
\end{equation}
The corresponding Young symmetrizer is: $c_S=\sum_{\sigma\in\fS_D} \sigma$. 
The projector on the $O(N)$ representation of traceless symmetric tensors is (see \cite[eq.~(4.21) and Prop.~4.2]{braueralg}):
\begin{equation}
    P_{D,N}^{(\lambda_S)} = \prod_{f=1}^{\lfloor\tfrac D2\rfloor} \Big( 1- \frac{A_D}{(N+2(D-f-1))f}  \Big) \;.
\end{equation}
This
is a restricted version 
of the universal traceless projector \eqref{eq: tracelessproj} and it removes the trace modes after restriction to symmetrized tensors.
As an element of the Brauer algebra $B_D(N)$, a propagator ($\boldsymbol{C}$ in Def.~\ref{def: model}) of a symmetric $O(N)$ tensor model is proportional to the projector:
\begin{equation}\label{eq: symmO}
    P_{D,N}^{(\lambda_S)} \frac{c_S}{D!} \;.
\end{equation}
The Brauer algebra acts on tensors by permuting and contracting their indices (see again subsection~\ref{sec: irreps}). 

To the $GL(N)$ representation of totally antisymmetric tensors $\bigwedge^D(V)$ of order $D$ is associated the Young diagram:
\begin{equation}\lambda_{\wedge}=\left.
\begin{ytableau}
\ \\ \\ \\ \none[\raisebox{-.3ex}{\vdots}] \\ \
\end{ytableau}\right\}{\text{\scriptsize length $D$}} \;.
\end{equation}
The corresponding Young symmetrizer is: $c_\wedge=\sum_{\sigma\in\fS_D} \sgn(\sigma) \sigma$.
A totally antisymmetric $O(N)$ tensor is automatically traceless, i.e.~$\bigwedge^D(V)$ is already an irreducible $O(N)$ representation.
Thus a propagator of an antisymmetric $O(N)$ tensor model, as an element of $B_D(N)$, is proportional to the projector:
\begin{equation}\label{eq: antisymmO}
    \frac{c_\wedge}{D!} \;,
\end{equation}
that acts by antisymmetrizing all indices.

\subsection{$Sp(N)$ tensor models}

In this case, the $N$-dimensional vector space $V$ is an odd (fermionic) real super-vector space and order $D$ tensors are bosonic if $D$ is an even integer and fermionic if $D$ is odd. 
Therefore the grading parameter now takes the value $\fb=1$.

The representation of the dual tensor model with $Sp(N)$ symmetry is obtained by transposing the Young diagram: $\lambda_\wedge^\prime=\lambda_S$. 
Therefore, the dual model to the symmetric traceless $O(N)$ tensor model is the antisymmetric traceless $Sp(N)$ tensor model. The projector onto this representation is given by:
\begin{equation}
    P_{D,-N}^{(\lambda_S)} = \prod_{f=1}^{\lfloor\tfrac D2\rfloor} \Big( 1- \frac{A_D}{(N+2(D-f-1))f}  \Big) \;,
\end{equation}
seen as an element of $B_D(-N)$ that differs from \eqref{eq: symmO} by the sign of $N$. 
Recall the difference in the action of $\beta$ when $\fb=1$ instead of $\fb=0$: the action of $\beta$ on tensors differs by a factor $\eta(\beta)=(-1)^{m}$, where 
$m=\text{minimal number of crossings in }\beta$. If $\beta$ is a permutation, we have $\eta(\beta)=\sgn(\beta)$ and thus, the Young symmetrizer $c_S\in B_D(-N)$ acts by antisymmetrizing the indices of a tensor, whereas $c_S\in B_D(N)$ acts by symmetrization. 

The dual model to the antisymmetric $O(N)$ tensor model contains tensors transforming in the symmetric representation of $Sp(N)$. Note that these tensors are also automatically traceless:
\begin{equation}
   \omega_{a_ia_j}  T^{a_1\dots a_i\dots a_j \dots a_D} =0 \;,
\end{equation}
because of the antisymmetry of the symplectic form. Thus, the projector is equal to \ref{eq: antisymmO}, but regarded as an element of $B_D(-N)$, and acts by symmetrizing the tensors. 
The diagrammatic (Feynman type) expansions of a  model and its dual contain exactly the same stranded graphs, but 
the amplitude of a stranded graph picks up a factor $N$ for each face in the $O(N)$ models, in the $Sp(N)$ models, each face contributes a factor $-N$. 
For example, the graph in Figure~\ref{fig:2colorstranded} has three faces and thus 
contributes as $N^3$ in an $O(N)$ model, but $-N^3$ in an $Sp(N)$ model.

\appendix

\section{Projector of tensors with irreducible symmetry}
\label{apendice}

Let us now consider an irreducible representation of $O(N)$ or $Sp(N)$ given by the Young tableau $\lambda$. The projector on this space of tensors is given by the product of the Young symmetrizer $c_\lambda$ with the traceless projector $\mathfrak{P}_D$:
\begin{equation}
\begin{split}
   P_{R}^{a_1 \ldots a_D a_{D+1} \ldots a_{2D}} &= (c_\lambda \cdot \mathfrak{P}_D)^{a_1 \ldots a_D a_{D+1}  \ldots a_{2D} } \\
   &=  \sum_{\substack{\alpha\text{ non-zero} \\ \text{eigenvalue of }A_D \\ \sigma \in P_\lambda , \tau \in Q_\lambda}}  \left(\sgn(\tau) \sigma \cdot \tau + \frac{-\sgn(\tau)}{\alpha}  \sigma \cdot \tau \cdot A_D\right)^{a_1\ldots a_D a_{D+1} \ldots a_{2D}}  
\end{split}
\end{equation}
The products of the elements $\sigma$, $\tau$ and $A_D$ of the Brauer algebra lead to $\phi$ and $\chi$, which are elements of $B_D((-1)^{\fb} N)$. We thus rewrite the terms present in $P_{R}$ as
\begin{equation}
\begin{split}
    \sum_{\substack{\alpha\text{ non-zero} \\ \text{eigenvalue of }A_D \\ \sigma \in P_\lambda \, , \tau \in Q_\lambda}} \sgn(\tau) \sigma \cdot \tau &= \sum_{\phi \in P_\lambda \cdot Q_\lambda} \gamma_\phi  \phi \; ,\\
    \sum_{\substack{\alpha\text{ non-zero} \\ \text{eigenvalue of }A_D \\ \sigma\in P_\lambda \, , \tau \in Q_\lambda}} \frac{-\sgn(\tau)}{\alpha}  \sigma \cdot \tau \cdot A_D &= \sum_{\chi \in P_\lambda \cdot Q_\lambda \cdot \mathbb{B}} \gamma_\chi \chi \; \text{,}
\end{split}
\end{equation}
where $\mathbb{B}$ is the set of elements $\beta_{ij}$ of the Brauer algebra (see \eqref{eq: betaij}), and $\gamma_\phi$ and resp. $\gamma_\chi$ are factors taking into account the fact that different products of $\sigma$ and $\tau$ may lead to the same $\phi$. We use then expression \eqref{eq:pairing_brauer} to write:

\begin{equation}
    P_{R}^{a_1 \ldots a_D a_{D+1} \ldots a_{2D}} = \sum_{\vec{\delta} \in \mathbb{M}_\delta} \gamma_\delta  \epsilon(\vec{\delta},\vM_{ref})^{\fb}
    \prod_{(i,j) \in \vec{\delta}} g^{a_i a_j} + \sum_{\vec{\omega} \in \mathbb{M}_\omega} \gamma_\tau \epsilon(\vec{\omega},\vM_{ref})^{\fb}
    \prod_{(i,j) \in \vec{\omega}} g^{a_i a_j} \, ,
\end{equation}
where the sum over the elements $\tau$ and $\delta$ of $B_D((-1)^{\fb}N)$ is replaced by a sum over their associated oriented pairings $\vec{\tau}$ and $\vec{\delta}$.

Thus, the projector $P_R$ is shown to be a particular case of the general projector \eqref{eq: propagator}.

\bigskip

\paragraph*{Acknowledgements.} 
The authors warmly acknowledge R\u azvan Gur\u au for useful discussions at various
stages of this research project.
T.~K., T.~M.~and A.~T.~have been partially supported by the ANR-20-CE48-0018 ``3DMaps'' grant and by the PHC Procope program "Combinatorics of random tensors". 
A.~T.~has been partially supported by the PN
23210101/2023 grant. H.~K.~has been supported by the Deutsche Forschungsgemeinschaft (DFG, German Research Foundation) under Germany's Excellence Strategy EXC--2181/1 -- 390900948 (the Heidelberg STRUCTURES Cluster of Excellence), and his mobilities were partially supported in the form of PPP France (DAAD). The authors further acknowledge support from the Institut Henri Poincaré (UAR 839 CNRS-Sorbonne Université), and LabEx CARMIN (ANR-10-LABX-59-01), where this project initiated, during the ``Quantum gravity, random geometry and holography'' trimester.

{\pagestyle{plain}
	\bibliography{references} 
	\addcontentsline{toc}{section}{References}
}


\end{document}